\documentclass[aos]{imsart}

\RequirePackage{amsthm,amsmath,amsfonts,amssymb}
\RequirePackage[colorlinks,citecolor=blue,urlcolor=blue]{hyperref}
\usepackage[capitalize]{cleveref}
\usepackage{mathrsfs} 

\usepackage{graphicx}
\graphicspath{{Images/}}
\usepackage{float}
\usepackage[figuresright]{rotating} 
\usepackage{pdfpages}
\usepackage{wrapfig}

\usepackage{array}
\usepackage{booktabs} 
\usepackage{enumerate}
\usepackage{multicol}
\usepackage{tikz}
\usepackage{algorithm}
\usepackage[noend]{algpseudocode}

\newcommand{\bH}{\mathbb{H}}
\newcommand{\bN}{\mathbb{N}}

\newcommand{\bR}{\mathbb{R}}
\newcommand{\bZ}{\mathbb{Z}}

\DeclareMathOperator*{\argmin}{\arg\min}

\usepackage{mathtools}

\DeclarePairedDelimiter\bra{\langle}{\rvert}
\DeclarePairedDelimiter\ket{\lvert}{\rangle}
\DeclarePairedDelimiterX\braket[2]{\langle}{\rangle}{#1\,\delimsize\vert\,\mathopen{}#2}

\startlocaldefs
\theoremstyle{plain}
\newtheorem{theorem}{Theorem}
\newtheorem{lemma}[theorem]{Lemma}
\newtheorem{corollary}[theorem]{Corollary}
\newtheorem{proposition}[theorem]{Proposition}

\newtheorem{definition}{Definition}

\newtheorem*{remark}{Remark}

\endlocaldefs

\begin{document}

\begin{frontmatter}
\title{Redheffer: Trig to Quantum Error Bounds}
\runtitle{Generalized of Redheffer Inequality}

\begin{aug}
\author[A]{\fnms{Ho} \snm{Yun}\ead[label=e1]{ho.yun@epfl.ch}}
\runauthor{H. Yun}
\address[A]{Ecole Polytechnique F\'ed\'erale de Lausanne, \printead{e1}}
\end{aug}

\begin{abstract}
In the existing literature, the Redheffer inequality is typically proven using mathematical induction. In this short paper, we present a straightforward proof of this inequality by leveraging trigonometric substitution. We then extend the Redheffer inequality by introducing an exponent factor, aiming for the sharpest possible refinement. Notably, when the exponent is 2, our findings have implications for quantum error correction in the context of quantum phase estimation.
\end{abstract}


\end{frontmatter}


\section{Introduction}
During my studies in the course on Quantum Algorithms taught by Professor Johannes Buchmann, I encountered an intriguing open problem that illuminated the quantum error estimates within phase estimation. In the process of solving this problem, a realization dawned upon me: the inequality at the heart of the quantum error estimates bore a resemblance to what is known as the Redheffer inequality. This serendipitous connection piqued my curiosity – the desire to delineate the critical value of $\alpha > 0$ for the fine-tuned Redheffer inequality within the bounded interval $x \in [0, 1/2]$:
\begin{equation}\label{Redheffer::goal}
\left(1 + 4x^{2} \right)^{1/\alpha} \cos(\pi x) \ge (1-4x^{2}).
\end{equation}

The sharpening of this inequality becomes apparent as $\alpha > 0$ increases, and the critical value of $\alpha > 0$ should be less than or equal to
\begin{equation}\label{opt::thres}
    \alpha_{T} := \lim_{x \nearrow 1/2} \frac{\log(1+4x^{2})}{\log \left((1-4 x^{2}) / \cos(\pi x) \right)} = \frac{\log 2}{\log (4/\pi)} (\approx 2.869).
\end{equation}

In the case where $\alpha = 1$, it is commonly referred to as the Redheffer inequality \cite{qi2006jordan}, and its proof typically relies on  mathematical induction in the literature. However, as can be seen in \cref{Redheffer::orig}, this proof is unnecessarily redundant for $\alpha = 1$. We proceed to demonstrate that a refined approach to mathematical induction effectively validates the inequality \eqref{Redheffer::goal} for $0< \alpha \le \log 2 / \log(21/16) (\approx 2.549)$.
Furthermore, we establish, without resorting to mathematical induction, that $\alpha_{T} (\approx 2.869)$ genuinely represents the sharpest bound. Finally, in the case where $\alpha=2$, our inequality provides the probability that best approximates the phase of an eigenvalue in the context of quantum phase estimation.

\section{Redheffer Inequality}
We present the proof of the original Redheffer inequality, which is notably simpler than the any other proofs available in the literature \cite{qi2006jordan, sandor2015inequality}.

\begin{proposition}[Redheffer]\label{Redheffer::orig}
For $x \in [0,1/2]$, $\left(1 + 4x^{2} \right) \cos(\pi x)  \ge (1-4x^{2})$, with equality holding if and only if $x=0$ or $x=1/2$.
\end{proposition}

\begin{proof}
Note that the map $\beta \in [0, \pi/2] \mapsto 2x=\tan(\beta/2) \in [0, 1/2]$ is bijective, and we obtain the following trigonometric substitution:
\begin{equation*}
    1-4x^{2}=1-\tan^{2} \left(\frac{\beta}{2} \right)=\frac{\cos \beta}{\cos^{2} \left(\frac{\beta}{2} \right)}, \quad
    1+4x^{2}=1+\tan^{2} \left(\frac{\beta}{2} \right)=\frac{1}{\cos^{2} \left(\frac{\beta}{2} \right)}.
\end{equation*}
Applying the change of variables, it remains to show that
\begin{equation*}
    \cos \left(\frac{\pi}{2} \tan \frac{\beta}{2} \right) \ge \cos \beta.
\end{equation*}
Since $\tan(\beta/2)$ is convex on $[0, \pi/2]$, it holds that
\begin{equation*}
    0 \le \frac{\pi}{2} \tan \frac{\beta}{2} \le \beta \le \frac{\pi}{2}, \quad \beta \in \left[0, \frac{\pi}{2} \right], 
\end{equation*}
which proves the inequality, with equality if and only if $\beta=0$ or $\beta=\pi/2$, i.e. $x=0$ or $x=1/2$.
\end{proof}

To achieve a more precise refinement of the Redheffer inequality, we rely on a well-known fact that for $x \in (-1, 1)$, the cosine function can be expressed as an infinite product:
\begin{equation*}
    \cos (\pi x)=\prod_{n=1}^{\infty} \left(1-\frac{4x^{2}}{(2n-1)^{2}} \right)
    =(1-4x^{2}) \lim_{n \rightarrow \infty} F_{n}(4x^{2}),
\end{equation*}
where the sequence of functions $\{F_{n}(y): n=2,3,\dots \}$ is defined as follows:
\begin{equation*}
    F_{n} (y) :=\prod_{k=2}^{n} \left(1-\frac{y}{(2k-1)^{2}} \right), \quad y \in [0, 1]. 
\end{equation*}
Applying $y=4x^{2}$, the generalized Redheffer inequality \eqref{Redheffer::goal} that we wish to prove becomes
\begin{align*}
    \lim_{n \rightarrow \infty} \left(1 + y \right)^{1/\alpha} F_{n}(y) \ge 1, \quad y \in [0, 1].    
\end{align*}

\subsection{Tuning through Induction}
We introduce a lemma that plays a critical role in our refinement:
\begin{proposition}\label{Redheffer::ind}
Let $\alpha >0$. If
\begin{equation}\label{lem::alpha}
    G_{n, \alpha}(y) := \left(1 + y \right)^{1/\alpha} F_{n} (y) - \left( 1+\frac{y}{4n-2} \right) \ge 0, \quad y \in [0, 1],
\end{equation}
holds for some $n=k \ge 2$, then it holds for any $n \ge k$. Consequently, we have
\begin{equation*}
    \left(1 + 4x^{2} \right)^{1/\alpha} \cos(\pi x) \ge (1-4x^{2}), \quad x \in [0, 1/2].
\end{equation*}
\end{proposition}
\begin{proof}
We prove it by mathematical induction on $n$. Assume that \eqref{lem::alpha} holds for $n=m \ge k$. For $n=m+1$, note that
\begin{equation*}
    \left(1 + y \right)^{1/\alpha} F_{m+1}(y) 
    =\left(1-\frac{y}{(2m+1)^{2}} \right) \cdot \left[ \left(1 + y \right)^{1/\alpha} F_{m}(y) \right],
\end{equation*}
hence it only remains to show that for $m \ge 2$ and $y \in [0,1]$,
\begin{equation*}    
    \left(1-\frac{y}{(2m+1)^{2}} \right)(1+a_{m} y) - (1+a_{m+1} y) \ge 0,
\end{equation*}
where $a_{m}=1/(4m-2)$. For any $y \in [0,1]$,
\begin{align*}
    &\left(1-\frac{y}{(2m+1)^{2}} \right)(1+a_{m} y) - (1+a_{m+1} y) \\
    &= \left(a_{m}-a_{m+1}-\frac{1}{(2m+1)^{2}} \right) y - \frac{a_{m}}{(2m+1)^{2}} y^{2} \\
    &\ge \left(a_{m}-a_{m+1}-\frac{1}{(2m+1)^{2}} - \frac{a_{m}}{(2m+1)^{2}} \right) y^{2} \\
    &= \left(\frac{1}{(2m-1)(2m+1)}-\frac{1}{(2m+1)^{2}} - \frac{1}{2(2m-1)(2m+1)^{2}} \right) y^{2} \\
    &=\frac{3}{2(2m-1)(2m+1)^{2}} y^{2} \ge \ 0,
\end{align*}
which proves \eqref{lem::alpha}. Therefore, for any $x \in [0, 1/2]$, it holds that
\begin{align*}
    \left(1 + 4x^{2} \right)^{1/\alpha} \cos(\pi x) 
    &=(1-4x^{2}) \left(1 + 4x^{2} \right)^{1/\alpha} \lim_{n \rightarrow \infty} F_{n}(4x^{2}) \\
    &\ge (1-4x^{2}) \lim_{n \rightarrow \infty} \left(1+\frac{2}{2n-1} x^{2} \right)
    =(1-4x^{2}).
\end{align*}
\end{proof}

For a fixed integer $n \ge 2$, the function $G_{n, \alpha}$ is decreasing in $\alpha > 0$. Hence, there exists $\alpha_{n}>0$, which we may call the $n$-threshold, such that
\begin{equation*}
    (0, \alpha_{n}] = \left\{\alpha>0 \, : \, G_{n, \alpha}(y) \ge 0, \, y \in [0, 1] \right\}.
\end{equation*}
Additionally, for a fixed $\alpha >0$, \cref{Redheffer::ind} indicates that if $G_{k, \alpha} \ge 0$, then $G_{n, \alpha} \ge 0$ for any $n \ge k$. In other words, the sequence $\left\{\alpha_{n}: n=2,3,\dots \right\}$ of thresholds monotonically increases. This observation highlights that we obtain a sharper bound for $\alpha > 0$ as we delay the initial step of mathematical induction.

\begin{proposition}
For any $0 \le \alpha < \alpha_{\infty} := \lim_{n \rightarrow \infty} \alpha_{n}$, we have
\begin{equation*}
    \left(1 + 4x^{2} \right)^{1/\alpha} \cos(\pi x) \ge (1-4x^{2}), \quad x \in [0, 1/2].
\end{equation*}
\end{proposition}
\begin{proof}
Given $0 \le \alpha < \alpha_{\infty}$, by \cref{Redheffer::ind}, there is some $n \ge 2$ such that $0 \le \alpha < \alpha_{n}$ so that \eqref{Redheffer::goal} holds.
\end{proof}

Note that the upper bound of $\alpha_{n}$ is given by
\begin{equation}\label{n::thres}
    \alpha_{n} \le \frac{\log 2}{\log (1+1/(4n-2))) - \log F_{n}(1)} =: \beta_{n},   
\end{equation}
since
\begin{equation*}
    G_{n, \alpha_{n}}(1) = 2^{1/\alpha} F_{n} (1) - \left( 1+\frac{1}{4n-2} \right) \ge 0.
\end{equation*}
If we aim to achieve the sharpest bound $\alpha_{T}$ through mathematical induction, we suspect that equality in \eqref{n::thres} should hold for any $n \ge 2$ since $\alpha_{\infty} \le \lim_{n \rightarrow \infty} \beta_{n} = \alpha_{T}$ by \cref{part::prod::deriv}. We show in the following proposition that $\alpha_{2} = \beta_{2}$. However, we are unable to prove if $\alpha_{n} = \beta_{n}$ for any $n \ge 2$, in general. This is because $G_{n, \beta_{n}}$ is no longer concave over $[0,1]$ for sufficiently large $n$, so the proof of the proposition below cannot be applied.
\begin{lemma}\label{lem::2::thres}
For any $7/9 < \alpha < 3$,
\begin{equation*}
    G_{2, \alpha}(y) = \left(1 + y \right)^{1/\alpha} \left(1-\frac{y}{9} \right) - \left( 1+\frac{y}{6} \right),
\end{equation*}
is concave on $[0, 1]$. Consequently, we have $\alpha_{2}=\beta_{2}= \log 2 / \log(21/16) (\approx 2.549)$. 
\end{lemma}
\begin{proof}
For a fixed $7/9 < \alpha < 3$,
\begin{align*}
    G_{2, \alpha} ''(y) :&= \left( \left(1 + y \right)^{1/\alpha} \right)'' \left(1-\frac{y}{9} \right) - 2 \left( \left(1 + y \right)^{1/\alpha} \right)' \left(1-\frac{y}{9} \right)' \\
    &= \frac{1}{\alpha} \left(1 + y \right)^{1/\alpha-2} \left[\left(\frac{1}{\alpha}-1 \right) \left(1-\frac{y}{9} \right) - \frac{2}{9 \alpha} (1+y)  \right] \\
    &= \frac{1}{\alpha} \left(1 + y \right)^{1/\alpha-2} \left[ -\frac{y}{9} \left(\frac{3}{\alpha}-1 \right) - \left( 1-\frac{7}{9 \alpha} \right) \right] <0, \quad y \in [0,1],  
\end{align*}
thus $G_{2, \alpha}$ is concave. If $\alpha> \log 2 / \log(21/16)$, then $G_{2, \alpha}(1)=8/9 \cdot 2^{1/\alpha}-7/6 <8/9 \cdot 21/16-7/6=0$, indicating that \eqref{lem::alpha} is not met. Conversely, if $\alpha \le \log 2 / \log(21/16)$, $G_{2, \alpha}$ is concave with $G_{2, \alpha_{2}}(0)=0, \, G_{2, \alpha}(1) \ge 0$, hence $G_{2, \alpha} \ge 0$.    
\end{proof}

\subsection{Sharpest Bound}
In this subsection, we establish that the value of $\alpha_{T}$, as defined in \eqref{opt::thres}, represents the most precise bound. The approach involves introducing a slight relaxation in $G_{n, \alpha}$, replacing $1+y/(4n-2)$ with $1$.

\begin{lemma}\label{part::prod::deriv}
For any $n \ge 2$, we have
\begin{equation*}
    \frac{F_{n}'(y)}{F_{n}(y)} = -\sum_{k=2}^{n} \frac{1}{(2k-1)^{2}-y}, \quad \frac{F_{n}'(1)}{F_{n}(1)}=-\frac{n-1}{4n}, \quad  \lim_{n \rightarrow \infty} F_{n}(1)=\frac{\pi}{4}.
\end{equation*}
\end{lemma}
\begin{proof}
Note that 
\begin{equation*}
    \log F_{n} (y) =\sum_{k=2}^{n} \log \left(\frac{(2k-1)^{2}-y}{(2k-1)^{2}} \right), \quad y \in [0, 1] 
\end{equation*}
Thus,
\begin{align*}
    \frac{F_{n}'(y)}{F_{n}(y)} = [\log F_{n}]' (y) = -\sum_{k=2}^{n} \frac{1}{(2k-1)^{2}-y}.
\end{align*}
If $y=1$, we obtain
\begin{align*}
    \frac{F_{n}'(1)}{F_{n}(1)}= - \sum_{k=2}^{n} \frac{1}{(2k-1)^{2}-1} = -\frac{1}{4} \sum_{k=2}^{n} \left(\frac{1}{k-1} - \frac{1}{k} \right) = -\frac{n-1}{4n}.
\end{align*}
Finally,
\begin{equation*}
    \lim_{n \rightarrow \infty} F_{n}(1)= \lim_{x \rightarrow 1/2} \frac{\cos (\pi x)}{1-4x^{2}} = \lim_{t \rightarrow 0} \frac{\sin (\pi t)}{4t (1-t)} =\frac{\pi}{4},
\end{equation*}
where we have put $x=1/2-t$.
\end{proof}

\begin{theorem}[Generalized Redheffer]\label{Redheffer::gen}
Let $\alpha > 0$.
\begin{equation}\label{Redheffer::alp}
    \left(1 + 4x^{2} \right)^{1/\alpha} \cos(\pi x)  \ge (1-4x^{2}), \quad x \in [0,1/2],
\end{equation} 
holds if and only if $0< \alpha \le \alpha_{T}$. In the case where $0< \alpha < \alpha_{T}$, equality is achieved if and only if $x=0$ or $x=1/2$.   
\end{theorem}

\begin{proof}
Let $\alpha > 0$ and $n \ge 2$. 
By \cref{part::prod::deriv}, it holds for any $y \in [0, 1]$ that
\begin{align*}
    [\left(1 + y \right)^{1/\alpha} F_{n} (y)]' 
    &= [\left(1 + y \right)^{1/\alpha}]' F_{n}(y) + \left(1 + y \right)^{1/\alpha} F_{n} '(y) \\
    &= \frac{1}{\alpha} \left(1 + y \right)^{1/\alpha-1} F_{n}(y) + \left(1 + y \right)^{1/\alpha} F_{n} '(y) \\
    &= \frac{1}{\alpha} \left(1 + y \right)^{1/\alpha-1} F_{n}(y) \left[1 + \alpha \left(1 + y \right) \frac{F_{n}'(y)}{F_{n}(y)} \right] \\
    &= \frac{1}{\alpha} \left(1 + y \right)^{1/\alpha-1} F_{n}(y) \left[1 - \alpha (1+y) \sum_{k=2}^{n} \frac{1}{(2k-1)^{2}-y} \right].
\end{align*}
Note that, for any $\alpha>0$,
\begin{equation*}
    \left[1 - \alpha (1+y) \sum_{k=2}^{n} \frac{1}{(2k-1)^{2}-y} \right] \searrow \left[1-\frac{n-1}{2n}\right] \quad \text{ as } y \nearrow 1
\end{equation*}
is strictly monotonely decreasing in $y$. Therefore, the function $\left(1 + y \right)^{1/\alpha} F_{n} (y)$ is either (i) strictly monotonely increasing or (ii) strictly monotonely increasing on $[0, y^{*}]$ and decreasing on $[y^{*}, 1]$ for some $y^{*} \in (0,1)$. In any case, $\left(1 + y \right)^{1/\alpha} F_{n} (y) \ge 1$ for any $y \in [0, 1]$ if and only if $2^{1/\alpha} F_{n}(1) \ge 1 \, \Longleftrightarrow \, \alpha \le - \log 2/\log F_{n}(1) =: \gamma_{n}$.

To show \eqref{Redheffer::alp}, note that $\gamma_{n} \nearrow \alpha_{T}$ as $n \rightarrow \infty$ by \cref{part::prod::deriv}. Thus, for any $0 \le \alpha < \alpha_{T}$, there is some $n \ge 2$ such that $\alpha < \gamma_{n}$, which yields
\begin{align*}
    \left(1 + 4x^{2} \right)^{1/\alpha} \cos(\pi x) 
    \ge (1-4x^{2}) \left(1 + 4x^{2} \right)^{1/\alpha} F_{n}(4x^{2})  \ge (1-4x^{2}), \quad x \in [0, 1/2].
\end{align*}
Indeed, \eqref{Redheffer::alp} remains valid in the case where $\alpha = \alpha_{T}$, because $\lim_{\alpha \nearrow \alpha_{T}} (1 + 4x^{2} )^{1/\alpha} = (1 + 4x^{2} )^{1/\alpha_{T}}$ for any $x \in [0, 1/2]$.

Finally, if $\alpha < \alpha_{T}$, choose $\tilde{\alpha}$ with $\alpha < \tilde{\alpha} < \alpha_{0}$. Then,
\begin{align*}
    \left(1 + 4x^{2} \right)^{1/\alpha} \cos(\pi x) 
    &=\left(1 + 4x^{2} \right)^{\left(1/\alpha-1/\tilde{\alpha}\right)} \left(1 + 4x^{2} \right)^{1/\tilde{\alpha}} \cos(\pi x) \\
    &\ge \left(1 + 4x^{2} \right)^{\left(1/\alpha-1/\tilde{\alpha}\right)} (1-4x^{2})
    \ge (1-4x^{2}),
\end{align*}
and equality holds if and only if $x=0$ or $x=1/2$.
\end{proof}

\begin{remark}
While the critical value of \cref{Redheffer::gen} is $\alpha_{T} = \log 2 /\log (4/\pi) (\approx 2.869)$, the threshold via mathematical induction that we achieved is $\alpha_{2}= \log 2 / \log(21/16) (\approx 2.549)$. 
\end{remark}

\begin{corollary}\label{QCthm}
For $\theta \in (0,1)$,
\begin{equation*}
    \sin^{2}(\pi \theta) \left(\frac{1}{\theta^{2}} + \frac{1}{(1-\theta)^{2}} \right) \ge 8,
\end{equation*}    
and equality holds if and only if $\theta=1/2$.
\end{corollary}

\begin{proof}
By substituting $x=\theta-1/2 \in (-1/2, 1/2)$, the inequality becomes
\begin{equation*}
    \cos^{2}(\pi x) \left(\frac{1}{(x-1/2)^{2}} + \frac{1}{(x+1/2)^{2}} \right) \ge 8,
\end{equation*}
or equivalently,
\begin{align*}
    &\cos^{2}(\pi x) \left( 4x^{2}+1 \right)=
    2 \cos^{2}(\pi x) \left( \left(x-\frac{1}{2} \right)^{2} + \left(x+\frac{1}{2} \right)^{2} \right) \ge \left(2x-1 \right)^{2} \left(2x+1 \right)^{2} \\
    &\Longleftrightarrow \quad \left(1 + 4x^{2} \right)^{1/2} \cos(\pi x)  \ge (1-4x^{2}).
\end{align*}
Since both hand sides are even functions in $x$, the result follows from \cref{Redheffer::gen} with $\alpha = 2$.
\end{proof}

\section{Quantum Error Correction}
We begin by presenting some fundamental postulates in quantum mechanics:
\begin{itemize}
\item \textbf{State space postulate}: A closed physical system is associated with a Hilbert space $\bH$, referred to as the state space of the system. A unit vector in $\bH$ completely describes the system.
\item \textbf{Measurement postulate}: Each observable is a Hermitian (self-adjoint) operator, denoted as $O=\sum_{\lambda_{i} \in \Lambda} \lambda_{i} \ket{\psi_{i}}\bra{\psi_{i}}$, acting on the state space $\bH$. When a particle initially has a quantum state $\ket{\phi}$, the probability of obtaining the measurement outcome $\ket{\psi_{i}}$ is given by $\mathbb{P}(\lambda_{i})=|\braket{\phi}{\psi_{i}}|^{2}$. Upon obtaining this outcome, the system instantly collapses into the eigenstate $\ket{\psi_{i}}$.
\end{itemize}

In quantum computing \cite{buchmann2016post}, the $n$-qubit state space $\bH_{n}$ is a $2^{n}$ dimensional complex vector space with an orthonormal basis $\{ \ket{\vec{x}} : \vec{x} \in \{0, 1\}^{n} \}$, commonly referred to as the computational basis. This basis can be represented using the tensor product, denoted as $\ket{\vec{x}} = \ket{x_{0}} \otimes \dots \ket{x_{n-1}}$. Consequently, we can view $\bH_{n}$ as the tensor product of $n$ copies of the single-qubit state space $\bH_{1} \cong \mathbb{C} \oplus \mathbb{C}$. The computational basis can be identified with the additive group of integers modulo $2^{n}$ through the canonical bijection:
\begin{equation*}
    \vec{x} = (x_{0}, \dots, x_{n-1}) \in \{0, 1\}^{n} \mapsto x = \sum_{j=0}^{n-1} x_{j} 2^{j} \in \bZ_{2^{n}},
\end{equation*}
so we use the notation $\vec{x} = \ket{x}_{n}$ henceforth. 

Let us consider the problem the problem of \textbf{quantum phase estimation}. Given a black-box unitary operator $U:\bH_{n} \rightarrow \bH_{n}$, our goal is to devise an efficient quantum circuit that can closely approximate an eigenvalue. This algorithm is frequently used as a subroutine in other quantum algorithms, such as Shor's algorithm. Since any eigenvalue has an absolute value $1$, it can be expressed as $\lambda=2^{2\pi \imath w}$ the phase $w \in \bR/\bZ$. However, representing any real number in $[0, 1)$ precisely using finite quantum gates is impossible. Thus, our task is to estimate the best possible proxy $x/2^{n}$ for the phase, where $x \in \bZ_{2^{n}}$.

\begin{definition}
Let $n \in \bN, w \in \bR$ and $x \in \bZ_{2^{n}}$. Define
\begin{equation*}
    \Delta(w, n, x)=\left\{ w-\frac{x}{2^{n}}-m : m=\argmin_{m' \in \bZ} \left| w-\frac{x}{2^{n}}-m' \right| \right\} \in [0, 1].
\end{equation*}
\end{definition}

Given a phase $w \in \bR/\bZ$, we associate the quantum state $\bH_{n}$ of $n$-qubit with a unit vector
\begin{equation*}
    \ket{\psi_{n}(w)}_{n} := \frac{1}{\sqrt{2^{n}}} \sum_{y=0}^{2^{n}-1} e^{2 \pi \imath w y} \ket{y}_{n}
    =\bigotimes_{j=0}^{n-1} \frac{\ket{0} + e^{2 \pi \imath 2^{j} w} \ket{1} }{\sqrt{2}}.
\end{equation*}
As the Schmidt rank of $\ket{\psi_{n}(w)}_{n}$ is $1$, it represents a non-entangled quantum state of $n$-qubits for any $w \in \bR$. 

The quantum Fourier transform (QFT) is a linear operator $QFT_{n}:\bH_{n} \rightarrow \bH_{n}$, which is the quantum analog of the discrete Fourier transform. It is uniquely determined by
\begin{equation*}
    QFT_{n} \ket{x}_{n} := \ket*{\psi_{n} \left( \frac{x}{2^{n}} \right)}_{n} = \frac{1}{\sqrt{2^{n}}} \sum_{y=0}^{2^{n}-1} e^{2 \pi \imath \frac{x}{2^{n}} y } \ket{y}_{n}, \quad x \in \bZ_{2^{n}},
\end{equation*}
and the inverse QFT is determined by
\begin{equation*}
    QFT^{-1}_{n} \ket{x}_{n} := \ket*{\psi_{n} \left( \frac{-x}{2^{n}} \right)}_{n} = \frac{1}{\sqrt{2^{n}}} \sum_{y=0}^{2^{n}-1} e^{-2 \pi \imath \frac{x}{2^{n}} y } \ket{y}_{n}, \quad x \in \bZ_{2^{n}}.
\end{equation*}
It can be easily shown that both operators are unitary.

QFT is a crucial quantum algorithm that can be efficiently implemented on a quantum computer. It can be decomposed into the product of simpler unitary matrices: for the discrete Fourier transform on $2^{n}$ amplitudes, QFT can be realized as a quantum circuit that consists of only $O(n^2)$ Hadamard gates and controlled phase shift gates  \cite{nielsen2010quantum}, as opposed to $O(n2^n)$ gates required for the classical discrete Fourier transform. The quantum Fourier transform operates on a computational basis, and its outcome is a sequence of probability amplitudes for all possible outcomes upon measurement, as dictated by the measurement postulate.

Returning to the quantum phase estimation problem discussed earlier, it is well known that when having access to the eigenstate that corresponds to the eigenvalue $\lambda=2^{2\pi \imath w}$ as an initial state, there exists a quantum circuit that yields the final state $QFT_{n}^{-1} \ket{\psi_{n}(w)}_{n}$ \cite{nielsen2010quantum}. The following theorem demonstrates that the outcome of the final state is close to the true phase $w$ with high probability.

\begin{proposition}
Let $n \in \bN, w \in \bR$ such that $2^{n}w \notin \bZ$. Let $p(x)$ denote the probability of measuring $x \in \bZ_{2^{n}}$, the outcome of measuring $QFT_{n}^{-1} \ket{\psi_{n}(w)}_{n}$ in the computational basis of $\bH_{n}$. Then with probability at least $8/\pi^{2}$, we have $| \Delta(w, n, x)| <2^{-n}$.
\end{proposition}
\begin{proof}
For the ease of notation, we denote $\Delta=\Delta(w, n, x)$. By the linearity of the inverse QFT,
\begin{align*}
    QFT_{n}^{-1} \ket{\psi_{n}(w)}_{n} =& \frac{1}{\sqrt{2^{n}}} \sum_{x=0}^{2^{n}-1} e^{2 \pi \imath w x} QFT_{n}^{-1} \ket{x}_{n} \\
    =& \frac{1}{2^{n}} \sum_{x=0}^{2^{n}-1} \sum_{y=0}^{2^{n}-1} e^{2 \pi \imath w y} e^{-2 \pi \imath \frac{x}{2^{n}} y } \ket{x}_{n} 
    = \sum_{x=0}^{2^{n}-1} \frac{1}{2^{n}}  \sum_{y=0}^{2^{n}-1} e^{2 \pi \imath \Delta y} \ket{x}_{n}.
\end{align*}
hence the probability of measuring $x$ is given by
\begin{equation*}
    p(x) = \frac{1}{2^{n}}  \sum_{y=0}^{2^{n}-1} e^{2 \pi \imath \Delta y} 
    =\frac{1 - e^{2 \pi \imath 2^{n} \Delta(w, n, x)}}{1 - e^{2 \pi \imath \Delta(w, n, x)}}
    =\frac{\sin^{2} (\pi 2^{n} \Delta)}{2^{2n} \sin^{2} (\pi \Delta)}
\end{equation*}
Therefore, by substituting $\theta = 2^{n} \Delta$,
the lower bound of the probability of measuring $x$ with $|2^{n} \Delta|<1$ is given by \cref{QCthm}:
\begin{align*}
    \frac{\sin^{2} (\pi 2^{n} \Delta)}{2^{2n} \sin^{2} (\pi \Delta)} + \frac{\sin^{2} (\pi 2^{n} (1-\Delta))}{2^{2n} \sin^{2} (\pi (1-\Delta))}
    =& \frac{\sin^{2} (\pi \theta)}{2^{2n} \sin^{2} (\pi \theta/2^{n})} + \frac{\sin^{2} (\pi \theta)}{2^{2n} \sin^{2} (\pi (1-\theta/2^{n}))} \\
    \ge& \frac{\sin^{2} (\pi \theta)}{\pi^{2}} \left(\frac{1}{\theta^{2}} + \frac{1}{(1-\theta)^{2}} \right) \ge \frac{8}{\pi^{2}}.
\end{align*}
\end{proof}

\section*{Acknowledgments}
I would like to thank Johannes Buchmann for invoking my curiosity in pursuit of solving the presented inequality, and to Yoav Zemel for giving me a clue to prove the inequality in a sharpest way possible.

\bibliographystyle{imsart-number}
\bibliography{bibliography}  

\end{document}